\documentclass{article}
\usepackage{spconf}
\usepackage[cmex10]{amsmath}
\usepackage{cite}
\usepackage{graphicx,epstopdf}
  \graphicspath{{../Figures/}}
  \DeclareGraphicsExtensions{.eps}
\usepackage{amssymb}
\usepackage{amsthm}
\newcounter{MYtempeqncnt}
\title{Secrecy Sum-Rates with Regularized Channel Inversion\\Precoding under Imperfect CSI at the Transmitter}
\name{Giovanni~Geraci$^{1,2}$,~Romain~Couillet$^3$,~Jinhong~Yuan$^1$,~M\'{e}rouane~Debbah$^4$~and~Iain~B.~Collings$^2$}
\address{$^1$ The University of New South Wales and $^2$CSIRO ICT Centre, Sydney, Australia \\
$^3$Dept. of Telecommunications and $^4$Alcatel-Lucent Chair on Flexible Radio,
Sup\'{e}lec, France}
\begin{document}
\maketitle
\begin{abstract}
In this paper, we study the performance of regularized channel inversion precoding in MISO broadcast channels with confidential messages under imperfect channel state information at the transmitter (CSIT). We obtain an approximation for the achievable secrecy sum-rate which is almost surely exact as the number of transmit antennas and the number of users grow to infinity in a fixed ratio. Simulations prove this anaylsis accurate even for finite-size systems. For FDD systems, we determine how the CSIT error must scale with the SNR, and we derive the number of feedback bits required to ensure a constant high-SNR rate gap to the case with perfect CSIT. For TDD systems, we study the optimum amount of channel training that maximizes the high-SNR secrecy sum-rate.
\end{abstract}

\section{Introduction}

Wireless multiuser communication is very susceptible to eavesdropping, and securing the transmitted information is critical. Although security has traditionally been ensured at the network layer by cryptographic schemes, this raises issues like key distribution and management, and high computational complexity. Moreover, these schemes rely on the unproven assumption that certain mathematical functions are hard to invert \cite{Mukherjee10Survey}. Therefore, tackling security for multiuser systems at the physical layer is of critical importance.

Physical layer security exploits the randomness inherent in noisy channels \cite{Wyner75,Csiszar78}. 
The secrecy capacity region of a broadcast channel with confidential messages (BCC), where the intended users can act maliciously as eavesdroppers, was studied in \cite{Liu08,Liu09Poor,Liu10}, but only for the case of two users. For a larger BCC with any number of malicious users, it was shown in \cite{GeraciISWCS11,Geraci11,GeraciWCNC12} that regularized channel inversion (RCI) precoding can achieve large secrecy sum-rates with low-complexity implementation, if the channel state information at the transmitter (CSIT) is perfect. However in real systems, CSIT must be acquired through channel feedback or estimation. This process can be challenging for practical time-varying scenarios, and the obtained information is inevitably imperfect.

In this paper, we study the secrecy sum-rates achievable by RCI precoding in the multiple-input single-output (MISO) BCC under imperfect CSIT. This work directly extends some of the analysis in \cite{Nguyen09,Wagner12} by requiring the transmitted messages to be kept confidential. Furthermore, this work generalizes the results in \cite{GeraciWCNC12}, where perfect CSIT was assumed. Our main contributions can be summarized as follows.
\begin{itemize}
\item We obtain a closed-form expression for the large-system secrecy sum-rate achievable by RCI precoding under imperfect CSIT. Simulations prove the analysis accurate even when applied to finite-size systems.
\item We study frequency division duplex (FDD) systems, and determine how the CSIT error must scale with the SNR, to maintain a given rate gap to the case with perfect CSIT. Under random vector quantization (RVQ), we find the minimum number of bits that each user must employ for channel quantization and feedback.
\item We study time division duplex (TDD) systems, and derive an expression for the secrecy sum-rate as a function of the amount of channel training. We obtain approximated expressions for the amount of channel training that maximizes the high-SNR secrecy sum-rate.
\end{itemize}
\section{System Model}
We consider a MISO BCC, where a base station (BS) with $M$ antennas simultaneously transmits $K$ independent confidential messages $\mathbf{u} = \left[u_1,\ldots,u_K \right]^{T}$, with E$[ \left|u_k\right|^2 ] =1$, to $K$ spatially dispersed single-antenna users, over a block fading channel. The transmitted signal is $\mathbf{x} = \left[x_1,\ldots,x_M \right]^{T}$. User $k$ receives
\vspace*{-0.2cm}
\begin{equation}
y_k=\sum_{j=1}^{M} h_{k,j}x_{j}+n_{k}
\label{eqn:MIMO_scalar}
\end{equation}
where $h_{k,j} \sim \mathcal{CN}(0,1)$ is the i.i.d. channel between the $j$-th transmit antenna element and the $k$-th user, and $n_{k} \sim \mathcal{CN}(0,\sigma^2)$ is the noise seen at the $k$-th receiver. 
The corresponding vector equation is $\mathbf{y}=\mathbf{Hx}+\mathbf{n}$, where $\mathbf{H} = \left[\mathbf{h}_1,\ldots,\mathbf{h}_K \right]$ is the $K \times M$ channel matrix. We define the downlink SNR as $\rho = 1/ \sigma ^2$ and impose E$[ \left\| \mathbf{x} \right\|^{2} ] =1$.

Only an imperfect estimate $\hat{\mathbf{H}}$ of the true channel $\mathbf{H}$ is available at the transmitter, which is modeled as
\vspace*{-0.1cm}
\begin{equation}
\mathbf{H} = \mathbf{\hat{H}} + \mathbf{E}
\label{eqn:CSI_model}
\end{equation}
where the error $\mathbf{E}$ is independent from $\mathbf{\hat{H}}$. The entries of $\mathbf{\hat{H}}$ and $\mathbf{E}$ are i.i.d. complex Gaussian random variables with zero mean and variances $1-\tau^2$ and $\tau^2$, respectively. The value of $\tau \in [0,1]$ depends on the quality and technique used for channel estimation. When $\tau=0$ the CSIT is perfectly known, whereas when $\tau=1$ no CSIT is available at all.

It is required that the BS securely transmits each confidential message $u_k$, ensuring that the unintended users receive no information. In general, the behavior of the users cannot be determined by the BS. As a worst-case scenario, in our system we assume that for each intended receiver $k$ the remaining $K-1$ users can cooperate to jointly eavesdrop on the message $u_k$.
For each user $k$, the alliance of the $K-1$ cooperating eavesdroppers is equivalent to a single eavesdropper with $K-1$ receive antennas, which is denoted by $\widetilde{k}$.

In this paper, we consider RCI precoding for the MISO BCC \cite{Peel05,Joham05}. RCI precoding is a linear scheme of particular interest because of its low-complexity \cite{Spencer04Magazine,Li10a}, and because it controls the amount of interference and information leakage to the unintended receivers \cite{GeraciISWCS11,Geraci11,GeraciWCNC12}. In RCI precoding, the transmitted vector $\mathbf{x}$ is obtained as $\mathbf{x} = \hat{\mathbf{W}}\mathbf{u}$, 
where \cite{Peel05,Joham05}
\begin{equation}
\hat{\mathbf{W}} = \left[\hat{\mathbf{w}}_1,\ldots,\hat{\mathbf{w}}_K \right] = \frac{1}{\sqrt{\gamma}} \left( \mathbf{\hat{H}}^H \mathbf{\hat{H}} + M \xi \mathbf{I} \right) ^{-1} \mathbf{\hat{H}}^H
\end{equation}
is the precoding matrix, $\gamma = \textrm{tr} \{ \mathbf{\hat{H}} ( \mathbf{\hat{H}}^H \mathbf{\hat{H}} + M \xi \mathbf{I} ) ^{-2} \mathbf{\hat{H}}^H \}$
is the power normalization constant, and $\xi$ is the regularization parameter that maximizes the large-system secrecy sum-rate under perfect CSIT, given in (\ref{eqn:xi_opt}) \cite{GeraciWCNC12}.
\begin{figure*}[!t]
\normalsize
\setcounter{MYtempeqncnt}{\value{equation}}
\setcounter{equation}{3}
\begin{equation}
\xi = \frac {-2\rho^2\left( 1 - \beta \right)^2 + 6\rho \beta + 2\beta^2 - 2 \left[ \beta \left( \rho+1 \right) -\rho \right] \cdot \sqrt{ \beta^2 \left[ \rho^2 + \rho + 1 \right] - \beta \left[ 2 \rho \left( \rho -1 \right) \right] + \rho^2 } } {6 \rho^2 \left( \beta + 2 \right) + 6 \rho \beta}
\label{eqn:xi_opt}
\end{equation}
\setcounter{equation}{\value{MYtempeqncnt}}
\addtocounter{equation}{1}
\hrulefill
\vspace*{-3pt}
\end{figure*}
\section{Large-System Analysis}

We now study the secrecy sum-rate of the RCI precoder in the large-system regime, where both the number of receivers $K$ and the number of transmit antennas $M$ approach infinity, with their ratio $\beta = K/M$ being held constant.

A secrecy sum-rate achievable by RCI precoding in the MISO BCC under perfect CSIT was obtained in \cite{Geraci11}, based on the work in \cite{Khisti10I}. Following a similar approach, a secrecy sum-rate achievable by RCI precoding in the presence of channel estimation error with variance $\tau^2$ is
\begin{equation}
R_{s} \!=\! \sum_{k=1}^{K} \left[ \log_2 \Big( 1 \!+\! \mathrm{SINR}_{k} \Big) \!-\! \log_2 \Big( 1 \!+\! \mathrm{SINR}_{\widetilde{k}} \Big) \right]^+,
\label{eqn:Rs}
\end{equation}
where $[\cdot]^+ \stackrel{\triangle}{=} \max(\cdot,0)$, and where $\mathrm{SINR}_{k}$ and $\mathrm{SINR}_{\tilde{k}}$ are the signal-to-interference-plus-noise ratios for the message $u_k$ at the intended receiver $k$ and the eavesdropper $\widetilde{k}$, respectively, given by
\begin{equation}
\mathrm{SINR}_{k} \!\!=\!\! \frac {\rho \left| \mathbf{h}_k^H \hat{\mathbf{w}}_k \right| ^2} {1 \!\!+\!\! \rho \sum_{j \neq k} {\left| \mathbf{h}_k^H \hat{\mathbf{w}}_j \right| ^2} }, \enspace
\mathrm{SINR}_{\widetilde{k}} \!\!=\!\! \rho \left\| \mathbf{H}_k \hat{\mathbf{w}}_{k} \right\| ^2
\end{equation}
and $\mathbf{H}_k$ is the matrix obtained from $\mathbf{H}$ by removing the $k$-th row.
We now give a deterministic equivalent for $R_s$.

\newtheorem{Theorem}{Theorem}
\begin{Theorem}
Let $\rho>0$, $\beta>0$, and let $R_s$ be the secrecy sum-rate in the presence of channel estimation error with variance $\tau^2$, defined in (\ref{eqn:Rs}). Define $\tilde{\rho} \stackrel{\triangle}{=} \frac{\rho\left(1-\tau^2\right)}{\rho \tau^2 + 1}$ and $\tilde{\xi} \stackrel{\triangle}{=} \frac{\xi}{1 - \tau^2}$. Then $\frac{1}{M} \left( R_s - R_s^{\circ} \right) \stackrel{M \rightarrow \infty}{\longrightarrow} 0$
almost surely, where $R_s^{\circ}$ is the large-system secrecy sum-rate under CSIT error, given by
\begin{equation}
R_s^{\circ} = K \left[ \log_2 \frac{1+ g ( \beta,\tilde{\xi} ) \frac{\tilde{\rho} + \frac{ \tilde{\xi} \tilde{\rho}}{\beta} \left[ 1 + g ( \beta,\tilde{\xi} ) \right]^2}{\tilde{\rho} + \left[ 1 + g ( \beta,\tilde{\xi} )  \right]^2 } }{1 + \rho \left[ \tau^2 + \frac{1-\tau^2}{\left(1 + g ( \beta,\tilde{\xi} )\right)^2} \right]} \right]^+
\label{eqn:Rs_deteq}
\end{equation}
with $g \left( \beta,\xi \right) \!=\!  \frac{1}{2} \left[ \textrm{sgn}(\xi) \sqrt{ \frac{\left(1\!-\!\beta \right)^2}{\xi^2}  \!+\!  \frac{2\left(1\!+\!\beta\right)}{\xi}  \!+\!  1} +  \frac{1\!-\!\beta}{\xi}  \!-\!  1 \right]$.
\end{Theorem}
\begin{proof}
From \cite{Nguyen09}, a deterministic equivalent for $\textrm{SINR}_k$ is
\begin{equation}
\textrm{SINR}_k^{\circ} = g ( \beta,\tilde{\xi} ) \frac{\tilde{\rho} + \frac{ \tilde{\xi} \tilde{\rho}}{\beta} \left[ 1 + g ( \beta,\tilde{\xi} ) \right]^2}{\tilde{\rho} + \left[ 1 + g ( \beta,\tilde{\xi} )  \right]^2 }.
\label{eqn:SINR_hat_deteq}
\end{equation}
Similarly, a deterministic equivalent for $\textrm{SINR}_{\tilde{k}}$ is given by
\begin{equation}
\textrm{SINR}_{\tilde{k}}^{\circ} = \rho \left[ \tau^2 + \frac{1-\tau^2}{\left(1 + g ( \beta,\tilde{\xi} )\right)^2} \right].
\label{eqn:SINR_tilde_hat_deteq}
\end{equation}
Theorem 1 then follows from (\ref{eqn:Rs}), (\ref{eqn:SINR_hat_deteq}), (\ref{eqn:SINR_tilde_hat_deteq}), and from the continuous mapping theorem \cite{BillingsleyBook}.
\end{proof}

When $\tau=0$, the value of $R_s^{\circ}$ in (\ref{eqn:Rs_deteq}) reduces to the large-system secrecy sum-rate under perfect CSIT, given by \cite{GeraciWCNC12}
\begin{equation}
\bar{R}_s^{\circ} = K \left[ \log_2 \frac{1+
g\left( \beta,\xi \right)
\frac{\rho + \frac{\rho\xi}{\beta} \left[ 1 + g\left( \beta,\xi \right) \right] ^2 
}{\rho + \left[ 1 + g\left( \beta,\xi \right) \right] ^2}}
{1+\frac{\rho}{ \left( 1+g\left( \beta,\xi \right) \right) ^2}} \right]^+.
\label{eqn:Rs_deteq_tau0}
\end{equation}

Fig. \ref{fig:Noisy_CSI_20dB} compares the secrecy sum-rate $R_s^{\circ}$ of the RCI precoder from the large-system analysis to the simulated ergodic secrecy sum-rate $R_s$ for finite $M$, under CSIT error $\tau=0.1$, and for different values of $\beta$. The values of $R_s^{\circ}$ and $R_s$ were obtained by (\ref{eqn:Rs_deteq}) and (\ref{eqn:Rs}), respectively. As expected, the accuracy of the deterministic equivalent increases as $M$ grows.

\begin{figure}
\centering
\includegraphics[width=\columnwidth]{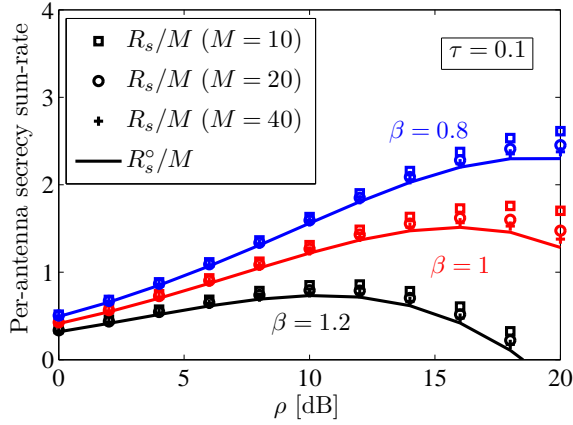}
\caption{Comparison between $R_s^{\circ}/M$ from (\ref{eqn:Rs_deteq}) and the simulated ergodic secrecy sum-rate $R_s/M$.}
\label{fig:Noisy_CSI_20dB}
\end{figure}

\section{Channel Feedback in FDD Systems}

In this section we consider the case of FDD systems, where users quantize their perfectly estimated channel vectors and send the quantization index back to the transmitter over a limited-rate channel. We assume that the channel magnitude is perfectly known to the transmitter, since it can be efficiently quantized, and that each channel direction is quantized using $B$ bits and random vector quantization (RVQ) \cite{Jindal06}. In RVQ, each user independently generates a random codebook with $2^B$ vectors, isotropically distributed on the $M$-dimensional unit sphere. RVQ generates a CSIT that follows the model in (\ref{eqn:CSI_model}), where the error $\tau^2$ can be upper bounded as \cite{Jindal06}
\begin{equation}
\tau^2 < 2^{-\frac{B}{M-1}}.
\label{eqn:bound_tau}
\end{equation}

In the following, we derive the necessary scaling of the channel estimation error $\tau^2$ to ensure a high-SNR constant rate gap of $\log_2 b$ bits to the case with perfect CSIT. A constant gap ensures that the multiplexing gain is not affected by the imperfect CSIT. We define the per-user rate gap $\Delta$ as 
\begin{equation}
\Delta \stackrel{\triangle}{=} \frac{\bar{R}_{s}^{\circ}-R_{s}^{\circ}}{K}
\end{equation}
where $\bar{R}_{s}^{\circ}$ and $R_{s}^{\circ}$ are the large-system secrecy sum-rates achieved by the RCI precoder under perfect CSIT and under CSIT distortion $\tau^2$, respectively.

It can be shown from (\ref{eqn:Rs_deteq_tau0}) that for $\beta>1$ the high-SNR secrecy sum-rate is zero irrespective of the distortion $\tau$, and this is also confirmed by the simulations in Fig. \ref{fig:Noisy_CSI_20dB}. This happens because in the worst-case scenario, the alliance of cooperating eavesdroppers can cancel the interference. Therefore when the number of users $K$ exceeds the number of antennas $M$, $\textrm{SINR}_k$ is interference limited, whereas $\textrm{SINR}_{\tilde{k}}$ is not, and $R_s$ tends to zero for all $\tau$. Since for $\beta>1$ both $\bar{R}_{s}^{\circ}$ and $R_{s}^{\circ}$ tend to zero in the high-SNR regime, in the remainder of the paper we will consider the case $\beta \leq 1$ as the most interesting. For this case, we obtain the following result.

\begin{figure}
\centering
\includegraphics[width=\columnwidth]{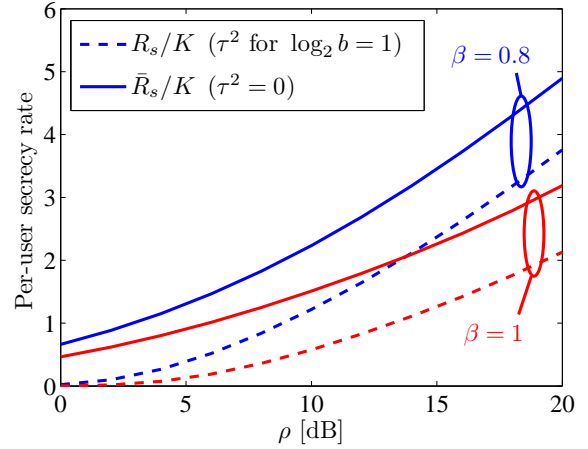}
\caption{Ergodic secrecy rate $\bar{R}_s/K$ under perfect CSIT vs $R_s/K$ under $\tau^2$ from Theorem 2, with $\log_2 b = 1$ bit.}
\label{fig:Limited_feedback}
\end{figure} 

\begin{Theorem}
In the large-system regime, for $\beta \leq 1$ and $b>1$, a CSIT distortion $\tau^2 = \frac{C}{\rho}$, with
\begin{equation}
C = \left\{ 
  \begin{array}{l l l}
     \frac{1}{2} \left( \sqrt{4b-3} - 1 \right) & \quad \textrm{if $\beta < 1$} \\
     \frac{2}{3} \left( \sqrt{3b-2} - 1 \right) & \quad \text{if $\beta = 1$} \\
  \end{array} \right.
\label{eqn:C}
\end{equation}
produces a high-SNR rate gap of $\log_2 b $ bits.
\end{Theorem}
\begin{proof}
Let $\mu \!\stackrel{\triangle}{=}\! \tau^2\rho \!+\! \frac{\tau^4\rho(1+g ( \beta,\xi ))^2}{1-\tau^2} \!+\! \frac{\tau^2(1+g ( \beta,\xi ))^2}{1-\tau^2}$. Then
\begin{equation}
\lim\limits_{\rho \rightarrow \infty} \Delta \!=\! \left\{ 
  \begin{array}{l l l}
     \!\!\lim\limits_{\rho \rightarrow \infty} \log_2 \left[ 1\!+\!\frac{\beta^2\mu}{4\rho(1-\beta)^2}  \right] \!=\! \log_2 b & \!\!\enspace \textrm{if $\beta < 1$} \\
     \!\!\lim\limits_{\rho \rightarrow \infty} \log_2 \left[  1\!+\!\frac{\mu}{4} \right] = \log_2 b & \!\!\enspace \text{if $\beta = 1$} \\
  \end{array} \right.
\end{equation}
\end{proof}

\newtheorem{Corollary}{Corollary}
\begin{Corollary}
In order to maintain a high-SNR secrecy rate offset of $\log_2 b$ bits per user in the large-system regime, it is sufficient to scale the number of feedback bits $B$ per user as
\begin{equation}
B \! \approx \! \left\{ 
  \begin{array}{l l l}
     \!\!\!\frac{M-1}{3} \rho_{\textrm{dB}} \!-\! \left( M\!\!-\!\!1 \right) [ \log_2 \left( \sqrt{4b-3} \!-\! 1 \right) \!-\! 1] & \enspace\!\!\! \textrm{if $\beta < 1$} \\
     \!\!\!\frac{M-1}{3} \rho_{\textrm{dB}} \!-\! \left( M\!\!-\!\!1 \right) [ \log_2 \frac{\sqrt{3b-2} - 1}{3} \!+\! 1] & \enspace\!\!\! \text{if $\beta = 1$} \\
  \end{array} \right.
\label{eqn:B}
\end{equation}
\end{Corollary}
\begin{proof}
The result follows from (\ref{eqn:bound_tau}) and Theorem 2.
\end{proof}

Fig. \ref{fig:Limited_feedback} shows the ergodic per-user secrecy rate $R_s/K$, achieved by the RCI precoder for $M=10$ in the presence of a channel estimation error that scales as in Theorem 2, with $\log_2 b = 1$ bit. This is compared to the ergodic rate $\bar{R}_s/K$, achieved by the same precoder under perfect CSIT ($\tau=0$). The simulations show that the desired secrecy rate gap of $1$ bit is approximately maintained at high SNR, thus confirming the claims made in Theorem 2. Fig. \ref{fig:Limited_feedback} also proves the analysis accurate even when applied to finite-size systems.
\section{Channel Training in TDD Systems}

We now consider a TDD system where uplink and downlink transmissions alternate on the same channel. The channel estimation at the transmitter is obtained from known pilot symbols sent by the users. Let $T$ be the channel coherence interval, i.e. the number of channel uses for which the channel is constant. The interval $T$ is divided into $T_t$ uses for uplink training and $T-T_t$ uses for the downlink transmission of data. The channel state information at the users is provided by a training phase in the downlink. However, a minimal amount of training is sufficient for this phase, and we can therefore neglect the overhead due to the downlink training \cite{MarzettaISIT09}.

Each user transmits the same number $T_t>K$ of orthogonal pilot symbols to the base station, which estimates all the $K$ channels simultaneously. The channel estimation error at the base station depends on the number $T_t$ as well as on the SNR $\rho_{ul}$ on the uplink channel, and it is given by \cite{Caire10}
\begin{equation}
\tau^2 = \frac{1}{1+T_t\rho_{ul}}.
\end{equation}
Hence we can write the secrecy sum-rate as a function of $T_t$
\begin{equation}
\hat{R}_s^{\circ} \!=\! \frac{T-T_t}{T} K \left[ \log_2 \frac{1 \!+\! g ( \beta,\tilde{\xi} ) \frac{\tilde{\rho} + \frac{ \tilde{\xi} \tilde{\rho}}{\beta} \left[ 1 + g ( \beta,\tilde{\xi} ) \right]^2}{\tilde{\rho} + \left[ 1 + g ( \beta,\tilde{\xi} )  \right]^2 } }{1 \!+\! \rho \left[ \tau^2 + \frac{1-\tau^2}{\left(1 + g ( \beta,\tilde{\xi} )\right)^2} \right]} \right]^+ \!\!
\label{eqn:Rs_TDD}
\end{equation}
where in the case of TDD systems we have
\begin{equation}
\tilde{\xi} = \frac{\xi(1+T_t\rho_{ul})}{T_t\rho_{ul}} \quad \textrm{and} \quad \tilde{\rho}=\frac{\rho T_t\rho_{ul}}{\rho+1+T_t\rho_{ul}}.	
\end{equation}
We now study the optimal value of the training interval $T_t$ for high SNR, in the case when the uplink and downlink SNR $\rho_{ul}$ and $\rho$ both grow with a finite ratio $c \stackrel{\triangle}{=} \rho/\rho_{ul}$.

\begin{Theorem}
In the large-system regime, let $\rho$, $\rho_{ul}$ be large with $c = \rho/\rho_{ul}$ constant. Then, an approximation of the secrecy sum-rate maximizing amount of channel training $T_t$ can be obtained as a solution of the equations
\begin{align}
&T_t^3q + T_t^2 (cq-Kc)\nonumber\\
&\enspace + T_t(c^2q+KcT-2Kc^2)+2Kc^2T=0,
\label{eqn:Tt1}\\
&T_t^34q + T_t^2 (4cq-4Kc) \nonumber\\
&\enspace + T_t(3c^2q+4KcT-6Kc^2)+6KTc^2=0,
\label{eqn:Tt2}
\end{align}
for $\beta<1$ and $\beta=1$, respectively, and with $q \stackrel{\triangle}{=} -\bar{R}_s^{\circ}\log2$.
\end{Theorem}
\begin{proof}
Rewrite $\hat{R}_s^{\circ} = \frac{T-T_t}{T} \left( \bar{R}_s^{\circ} - K \Delta \right)$, where
\begin{equation}
\lim\limits_{\rho,\rho_{ul} \rightarrow \infty} \Delta = \left\{ 
  \begin{array}{l l l}
    \log_2\left(1+\frac{c^2}{T_t^2}+\frac{c}{T_t} \right) & \enspace\!\!\! \textrm{for $\beta < 1$} \\
    \log_2\left(1+\frac{3c^2}{4T_t^2}+\frac{c}{T_t} \right) & \enspace\!\!\! \text{for $\beta = 1$} \\
  \end{array} \right. .
\end{equation}
Then (\ref{eqn:Tt1}) and (\ref{eqn:Tt2}) can be obtained by setting $\partial \hat{R}_s^{\circ} / \partial T_t = 0$ and after further high-SNR approximations.
\end{proof}

Fig. \ref{fig:TDD} shows the simulated optimal relative amount of training $T_t/T$ versus the downlink SNR $\rho$, for a system with $M=K=10$, and $c=10$. This is compared to the high-SNR approximations obtained from Theorem 3. Three sets of curves are shown, each for a different coherence time $T$. The figure shows that for increasing SNR, the channel estimation becomes more accurate, and less resources should be allocated to channel training. At high SNR, the optimal amount of training converges to the value predicted by the analysis and provided in Theorem 3.

\begin{figure}
\centering
\includegraphics[width=\columnwidth]{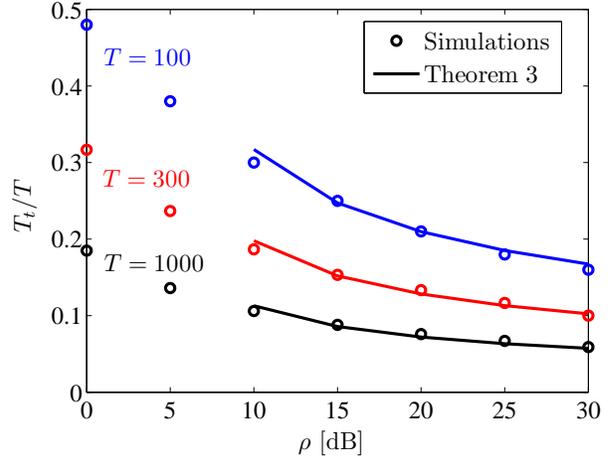}
\caption{Optimal relative amount of training $T_t/T$ vs high-SNR approximation, for $M=K=10$ and $c=\rho/\rho_{ul}=10$.}
\label{fig:TDD}
\end{figure} 
\section{Conclusion}
In this paper, we considered regularized channel inversion precoding for the MISO broadcast channel with confidential messages, under imperfect channel state information at the transmitter. For this system set-up, we first studied the precoder in the large-system regime, and obtained a deterministic equivalent for the achievable secrecy sum-rate. This analysis proved to be accurate even for finite-size systems. We then determined the amount of feedback required to the users in an FDD system to ensure a constant high-SNR rate gap to the case with perfect channel state information, so that the multiplexing gain is not affected. We further considered the case of TDD systems and studied the optimum amount of channel training that maximizes the high-SNR secrecy sum-rate.
\bibliographystyle{IEEEbib}
\bibliography{IEEEabrv,Bib_Giovanni}
\end{document}